\documentclass[11pt,a4paper]{article}%

\usepackage{fullpage}
\usepackage{amsmath}
\usepackage{graphicx}%
\usepackage[font=small]{caption}
\usepackage{subcaption}
\usepackage{amsfonts}%
\usepackage{amssymb}
\usepackage{amsthm}
\usepackage{mathrsfs}
\usepackage{cite}
\newtheorem{theorem}{Theorem}

\newtheorem{example}[theorem]{Example}

\newtheorem{lemma}[theorem]{Lemma}

\newtheorem{proposition}[theorem]{Proposition}
\newtheorem{remark}[theorem]{Remark}

\begin{document}

\title{Tangential Interpolatory Projection for Model Reduction of Linear Quantum Stochastic Systems\thanks{Research supported by the Australian Research Council}}
\author{Onvaree~Techakesari and Hendra~I.~Nurdin 
\footnote{The authors are with the School of Electrical Engineering and Telecommunications, The University of New South Wales (UNSW), 
Sydney NSW 2052, Australia. Email: o.techakesari@unsw.edu.au \& h.nurdin@unsw.edu.au}}

\maketitle \thispagestyle{empty}

\begin{abstract}
This paper presents a model reduction method for the class of linear quantum stochastic systems often encountered in quantum optics and their related fields. The approach is proposed on the basis of an interpolatory projection ensuring that specific input-output responses of the original and the reduced-order systems are matched at multiple selected points (or frequencies). Importantly, the physical realizability property of the original quantum system imposed by the law of quantum mechanics is preserved under our tangential interpolatory projection. An error bound is established for the proposed model reduction method and an avenue to select interpolation points is proposed. A passivity preserving model reduction method is also presented. 
Examples of both active and passive systems are provided to illustrate the merits of our proposed approach.
\end{abstract}

\section{Introduction}

Over the past few decades, considerable interest has been drawn to quantum systems involving open harmonic oscillators linearly coupled to one another and to external Gaussian fields, especially in areas such as quantum optics, optomechanics, and superconducting circuits. This type of systems is used to model quantum devices such as optical cavities, mesoscopic mechanical resonators, optical and superconducting parametric amplifiers, and gradient echo quantum memories (GEM); e.g.. see \cite{GZ04,HM12,KAKKCSL13,HCHJ13}. A number of applications of linear quantum stochastic systems have been theoretically proposed or experimentally demonstrated in the literature. In particular, they can serve as coherent feedback controllers \cite{JNP06,NJP07b}, i.e., feedback controllers that are themselves quantum systems. In this context,  they have been shown to be theoretically effective for cooling of an optomechanical resonator \cite{HM12}, can modify the characteristics of squeezed light produced experimentally by an optical parametric oscillator  (OPO) \cite{CTSAM13}, and, in the setting of microwave superconducting circuits, a linear quantum stochastic system in the form of a Josephson parametric amplifier (JPA) operated in the linear regime has been experimentally demonstrated to be able to rapidly reshape the dynamics of a superconducting electromechanical circuit (EMC) \cite{KAKKCSL13}. Linear quantum stochastic systems can also be used as optical filters for various input signals, including non-Gaussian input signals like single photon and multi-photon states. As filters they can be used to modify the wavepacket shape of single and multi-photon sources \cite{ZJ13,Zhang14}. Also, linear quantum stochastic  systems can dissipatively generate Gaussian cluster states \cite{KY12} as an important component of continuous-variable one way quantum computers \cite{MLGWRN06}. They can also be exploited for classical signal processing on quantum devices, such as in light processors where photons transport information between cores on a chip and between chips \cite{BKSWW07}. 

The linearly coupled open harmonic oscillators can be completely represented in Heisenberg picture of quantum mechanics by a class of linear quantum stochastic systems described using quantum stochastic differential equations (QSDEs) and subsequently, a quartet of matrices $(A,B,C,D)$ (analogous to classical non-quantum linear systems) \cite{BE08,JNP06,GGY08,GJN10,GZ15,NJD08}. However, unlike many classical non-quantum systems, these matrices cannot be arbitrary and they are not independent of one another due to the restrictions imposed by the laws of quantum mechanics. In fact, $A,B$, and $C$ must satisfy certain equality constraints while $D$ must be of a specific form. These constraints are referred to as the {\it physical realizability} conditions \cite{JNP06}. The class of linear quantum stochastic systems provides us with a tool to develop fundamental principles in quantum control theory in a similar way that classical non-quantum linear control systems have facilitated the advancement of classical systems and control theory.

In many linear quantum control problems, the input-output relation of the system (or controller) described by its transfer function is more important than its realization described by the matrices $(A,B,C,D)$ e.g., see \cite{YK03a,YK03b,JNP06,GGY08,Mab08,NJP07b,MP11a}. 
Unfortunately, the system transfer functions may be complex involving a large degree of freedom and hence, physically implementing systems with such transfer functions can be a challenging task. In these situations, it is more appealing to construct an approximating system with a lower degree of freedom such that the transfer functions of the original and the approximating systems are closely matched. The process of constructing a lower order approximating system is known as model reduction. The contribution of this paper lies in this model reduction problem.

A challenge facing model reduction in quantum setting is the retention of the physical realizability property. In \cite{BvHS07,GNW10,P13,VP11}, singular perturbation methods (also known as adiabatic elimination) were studied for reduction of linear quantum systems comprising subsystems that evolve at well-separated time-scales. In \cite{Pet12}, an eigenvalue truncation method was proposed for a sub-class of completely passive systems with distinct poles. More recently, an adaption of the well-known balanced truncation method was proposed in \cite{Nurd14} for linear quantum stochastic systems whose controllability and observability Gramians satisfy some commutation conditions. 

As the existing model reduction methods are limited to sub-classes of systems with specific properties, this paper presents a model reduction approach that allows an approximation of a more general linear quantum stochastic system. In particular, we propose a model reduction method on the basis of the tangential interpolatory projection method introduced in \cite{GVV04}. This approach constructs an approximating system such that its transfer function interpolates the transfer function of the original system at multiple points along some tangent directions. Importantly, we show that the reduced system is physically realizable. An error bound is established for the proposed method, and an avenue to select interpolation points and tangent directions are presented. 

The remainder of this paper is structured as follows. In Section \ref{sec:prelim}, we define the class of linear quantum stochastic systems under consideration. The tangential interpolatory projection model reduction approach is proposed in Section \ref{sec:modred}. Section \ref{sec:passive} presents a passivity preserving model reduction method. Concluding remarks are then provided in Section \ref{sec:conclusion}.

\section{Preliminaries} \label{sec:prelim}

\subsection{Notation}
We will use the following notation: $\imath=\sqrt{-1}$, $^*$ denotes the adjoint of a linear operator
as well as the conjugate of a complex number. If $A=[a_{jk}]$ then $A^{\#}=[a_{jk}^*]$, and $A^{\dag}=(A^{\#})^{\top}$, where $(\cdot)^{\top}$ denotes matrix transposition.  $\Re\{A\}=(A+A^{\#})/2$ and $\Im\{A\}=\frac{1}{2\imath}(A-A^{\#})$.
We denote the identity matrix by $I$ whenever its size can be
inferred from context and use $I_{n}$ to denote an $n \times n$
identity matrix. Similarly, $0_{m \times n}$ denotes  a $m \times n$ matrix with zero
entries but we drop the subscript when its dimension  can be determined from context. We use
${\rm diag}(M_1,M_2,\ldots,M_n)$  to denote a block diagonal matrix with
square matrices $M_1,M_2,\ldots,M_n$ on its diagonal, and ${\rm diag}_{n}(M)$
denotes a block diagonal matrix with the square matrix $M$ appearing on its diagonal blocks $n$ times. 
Also, we let \small $\mathbb{J}=\left[\begin{array}{rr}0 & 1\\-1&0\end{array}\right]$\normalsize and $\mathbb{J}_n=I_{n} \otimes \mathbb{J}={\rm diag}_n(\mathbb{J})$. We use $\mathbb{N}$ to denote the set of natural numbers.
We use $e_i = [0,0,\ldots,0,1,0,\ldots,0]^\top$ to denote an indicator vector with one in the $i$-th position and zero elsewhere.
For a subspace ${\cal H}$ of a vector space, we write $P_{\cal H}$ to denote an orthogonal projection onto ${\cal H}$.
We also write ${\cal H}^\perp$ to denote the orthogonal complement of ${\cal H}$.

\subsection{The class of linear quantum stochastic systems in quadrature form}

A {\em linear quantum stochastic system} \cite{BE08,JNP06,NJD08} is an open Markov quantum system involving $n$ independent quantum harmonic oscillators coupled to $m$ independent external continuous-mode bosonic fields. 
Let $x=(q_1,p_1,q_2,p_2,\ldots,q_n,p_n)^\top$ denote a vector of  the canonical position and momentum operators of a {\em $n$ degree of freedom quantum harmonic oscillator} satisfying the canonical commutation relations (CCR)  $xx^\top -(xx^\top )^\top =2\imath \mathbb{J}_n$. That is, $q_j$ and $p_j$ are position and momentum operators of $j$-th harmonic oscillator, respectively.
Let $Y(t)=(Y_1(t),\ldots,Y_m(t))^{\top}$ denote a vector of continuous-mode bosonic {\em output fields} that results from the interaction of the quantum harmonic oscillators and the incoming continuous-mode bosonic quantum fields in the $m$-dimensional vector
$\mathcal{A}(t)$. Note that the dynamics of $x(t)$ is linear, and
$Y(t)$ depends linearly on $x(s)$ with $0 \leq s \leq t$. 
Following \cite{JNP06}, the dynamics of a linear quantum stochastic system can be described in quadrature form as
\begin{align}
dx(t)&=Ax(t)dt+Bdw(t);\quad x(0)=x. \nonumber\\
dy(t)&= C x(t)dt+ D dw(t), \label{eq:qsde-out-quad}
\end{align}
where $A \in \mathbb{R}^{2n\times 2n}$, $B \in \mathbb{R}^{2n\times 2m}$, $C \in \mathbb{R}^{2m\times 2n}$, and $D \in \mathbb{R}^{2m\times 2m}$,
with
\begin{align*}
 w(t)
 &= 2\big(\Re\{\mathcal{A}_1(t)\},\Im\{\mathcal{A} _1(t)\},\Re\{\mathcal{A} _2(t)\},\Im\{\mathcal{A} _2(t)\},\ldots, \Re\{\mathcal{A} _m(t)\},\Im\{\mathcal{A} _m(t)\}\big)^{\top}, \\
 y(t)
 &=2\big(\Re\{Y_1(t)\},\Im\{Y_1(t)\},\Re\{Y_2(t)\},\Im\{Y_2(t)\}, \ldots, \Re\{Y_m(t)\}, \Im\{Y_m(t)\}\big)^{\top}.
\end{align*}
Here, $w(t)$ (input fields in quadrature form) is taken to be in a vacuum state where it satisfies the It\^{o} relationship $dw(t)dw(t)^{\top} = (I+\imath \mathbb{J}_m)dt$; see \cite{JNP06}. Note that in this form it follows that $D$ is a real unitary symplectic matrix. That is, it is both unitary (i.e., $DD^{\top}=D^{\top} D=I$) and symplectic (i.e., $D \mathbb{J}_m D^{\top} =\mathbb{J}_m$). However, in the most general case, $D$ can be generalized to a symplectic matrix that represents a quantum network that includes ideal squeezing devices acting on the incoming field $w(t)$ before interacting with the system \cite{GJN10,NJD08}.  
In general one may not be interested in all outputs of the system but only in a subset of them, see, e.g., \cite{JNP06}. That is, one is often only interested in certain pairs of the output field quadratures in $y(t)$. Thus, in the most general scenario, $y(t)$ can have an even dimension $2\ell <2m$. 

Unlike classical non-quantum systems, the matrices $A$, $B$, $C$, $D$ of a linear quantum stochastic system cannot be arbitrary and are not independent of one another. In fact, for the system to be physically realizable \cite{JNP06,NJD08}, meaning it represents a meaningful physical system, they must satisfy the constraints (see \cite{JNP06,NJD08,GJN10,Nurd14}):
\begin{eqnarray}
&&A\mathbb{J}_n + \mathbb{J}_n A^{\top} + B\mathbb{J}_mB^{\top}=0, \label{eq:pr-1}\\
&& \mathbb{J}_n  C^{\top} +B\mathbb{J}_mD^{\top}=0, \label{eq:pr-2}\\
&&D \mathbb{J}_m D^{\top} =  \mathbb{J}_{\ell}. \label{eq:pr-3}
\end{eqnarray}

\section{Tangential Interpolatory Projection for Model Reduction} \label{sec:modred}

Interpolatory projection model reduction framework was developed by De Villemagne and Skelton in \cite{DS85} to construct a reduced-order system such that its transfer function interpolates the transfer function of the original system at a single point.
Limitations of this single-interpolation-point approach were discussed in \cite{GGV96,Gphd97} including the loss of approximation accuracy away from the single interpolation point.
For this reason, multi-point interpolatory projection framework was introduced. In particular, two key tangential interpolatory projection approaches were proposed in \cite{GVV04}: the {\em left} and the {\em right} tangential interpolatory projections.
Let $\Xi(s)$ and $\Xi_r(s)$ be the transfer functions of the original and the reduced-order systems, respectively.
Given a set of $2r$ interpolation points $\{\sigma_1,\sigma_2,\ldots,\sigma_{2r}\}$ and $2r$ left (or output) tangent directions  $\{ \mu_1,\mu_2,\ldots,\mu_{2r}\} \subset \mathbb{C}^{2\ell}$, the left tangential interpolatory projection approach ensures that
\begin{equation}
 \mu^\dagger_i\Xi_r(\sigma_i) = \mu^\dagger_i\Xi(\sigma_i) \label{eq:interpol-cond-left}
\end{equation}
for all $i = 1,2,\ldots,2r$.
Similarly, the right interpolatory approach ensures that
\begin{equation}
 \Xi_r(\sigma_i)\nu_i = \Xi(\sigma_i)\nu_i \label{eq:interpol-cond-right}
\end{equation}
for all $i = 1,2,\ldots,2r$,
where $\nu_1,\nu_2,\ldots,\nu_{2r} \in \mathbb{C}^{2m}$ are the right (or input) tangent directions.

We stress that we are interested in constructing a reduced system while meeting either \eqref{eq:interpol-cond-left} or \eqref{eq:interpol-cond-right}. In general, one might consider the left tangential interpolatory projection approach (instead of the right) because the left tangent directions belong to a smaller complex space $\mathbb{C}^{2\ell}$ than (or at most equal to) the space $\mathbb{C}^{2m}$ where each right tangent direction lives in (since $\ell \leq m$ for the class of linear quantum stochastic systems). The smaller space reduces the choices of tangent directions and may simplify the selection of appropriate tangent directions. Intuitively, this may also help in mitigating the effects the use of inadequate tangent directions may have on model approximation error.
However, in some situations, it may be more appropriate to consider the right tangential interpolatory projection e.g., when one wants to place more emphasis or weights on the input components that are of interest. For this reason, we will present both the left and the right tangential interpolatory projection approaches.

\subsection{Petrov-Galerkin Projection Approximation}

The tangential interpolatory projection method can be achieved by constructing a reduced-order model of a linear quantum stochastic system via Petrov-Galerkin projection approximation. That is, given a full-order linear quantum stochastic system in quadrature form \eqref{eq:qsde-out-quad}, we construct a reduced-order linear quantum stochastic system of the form
\begin{align}
 dx_r(t) &= A_r x_r(t) dt + B_r dw(t); \quad x_r(0) = W_q^\top x \nonumber \\
 dy_r(t) &= C_r x_r(t) dt + D_r dw(t)  \label{eq:qsde-quad-reduced}
\end{align}
where $A_r = W_q^\top A V_q$, $B_r = W_q^\top B$, $C_r = CV_q$, $D_r = D$.
Here, $W_q \in \mathbb{R}^{2n \times 2r}$ and $V_q \in \mathbb{R}^{2n \times 2r}$ are the left and right projection matrices, respectively, which satisfy the condition $W_q^\top V_q = I$. 

Let $\widehat{\Xi}(s) = C\left(sI - A\right)^{-1} B$ and $\widehat{\Xi}_r(s) = C_r\left(sI - A_r\right)^{-1} B_r$.
Since we consider $D_r = D$, it can be seen that the left tangential interpolation condition \eqref{eq:interpol-cond-left} can be achieved by meeting a simpler condition $\mu_i^\dagger\widehat{\Xi}(\sigma_i) = \mu_i^\dagger\widehat{\Xi}_r(\sigma_i)$ for all $i = 1,2,\ldots,2r$. Similarly, the right tangential interpolation condition \eqref{eq:interpol-cond-right} is attained when $\widehat{\Xi}(\sigma_i)\nu_i = \widehat{\Xi}_r(\sigma_i)\nu_i$ is satisfied for all $i = 1,2,\ldots,2r$.
On the basis of these simplified interpolation conditions, we can now re-state an important result on tangential interpolatory projection in terms of our linear quantum stochastic systems in quadrature form and Petrov-Galerkin projection approximation.

\begin{theorem} \cite[Theorem 1]{ABG10} \label{thm:interpol-quad}
Given a full-order linear quantum stochastic system written in quadrature form \eqref{eq:qsde-out-quad} with the transfer function $\Xi(s)$, consider a reduced-order model of the form \eqref{eq:qsde-quad-reduced} constructed through Petrov-Galerkin projection approximation with the transfer function $\Xi_r(s)$. Then we have that
\begin{enumerate} 
\item[(i)] the left interpolation condition \eqref{eq:interpol-cond-left} holds if $W_q$ has full rank and $\left(\mu_i^\dagger C(\sigma_i I - A)^{-1}\right)^\dagger \in {\rm range}(W_q)$ for all $i = 1,2,\ldots,2r$.

\item[(ii)] the right interpolation condition \eqref{eq:interpol-cond-right} holds if $V_q$ has full rank and $(\sigma_i I - A)^{-1}B\nu_i \in {\rm range}(V_q)$ for all $i = 1,2,\ldots,2r$.
\end{enumerate}
\end{theorem}

For many classical (non-quantum) systems, it is possible to choose $W_q$ and $V_q$ independently so that both the left and right interpolation conditions \eqref{eq:interpol-cond-left}-\eqref{eq:interpol-cond-right} are satisfied. 
Unfortunately for quantum systems, due to the law of quantum mechanics, the projection matrices $W_q$ and $V_q$ cannot be independently selected. Hence, using the above result, we will present avenues to choose $W_q$ and $V_q$ such that either the left interpolation condition \eqref{eq:interpol-cond-left} or the right interpolation condition \eqref{eq:interpol-cond-right} is satisfied, while ensuring that the physical realizability property of a linear quantum stochastic system is preserved.
However, before presenting our physical realizability preserving model reduction method, let us recall the following well-known result.

\begin{lemma} \cite{HJ85}\label{lem:skewSymmetric}
For any full-rank (non-singular) real skew-symmetric matrix $\Theta \in \mathbb{R}^{2n\times 2n}$ (i.e., $\Theta = -\Theta^\top$), there exists a non-singular (full-rank) real matrix $T \in \mathbb{R}^{2n\times 2n}$ such that $T \Theta T^\top = \mathbb{J}_n$.
\end{lemma}

\subsection{Proposed Physical Realizability Preserving Model Reduction}

Physical realizability is an important property of a quantum system. In constructing a reduced-order model of a linear quantum stochastic system, it is important to ensure that the reduced system represents a meaningful physical quantum system. 
As previously mentioned, the physically realizability property places some restrictions on the projection matrices $W_q$ and $V_q$.
Motivated by the method presented in \cite{DS85}, we will now present avenues to choose $W_q$ and $V_q$ so that the reduced system is guaranteed to be physically realizable.

\subsubsection{Left Tangential Interpolatory Projection}

To achieve the left tangential interpolation condition \eqref{eq:interpol-cond-left}, consider a subspace
\begin{align}
 {\cal W}_q 
 &\triangleq {\rm span}\left\{ \left(\mu_i^\dagger C(\sigma_iI-A)^{-1}\right)^\dagger\right\}_{i = 1,2,\ldots,2r}
 \label{eq:Wsubspace-quad}
\end{align}
where interpolation points, $\{\sigma_1,\sigma_2,\ldots,\sigma_{2r}\}$, and the left tangent directions, $\{\mu_1,\mu_2,\ldots,\mu_{2r}\}$, are chosen such that
\begin{enumerate}
 \item[(i)] $(\sigma_i I - A)$ is invertible for all $i = 1,2,\ldots,2r$,
 \item[(ii)] $\mu_i \neq 0$ for all $i = 1,2,\ldots,2r$,
 \item[(iii)] the subspace ${\cal W}_q$ has dimension $2r$, and
 \item[(iv)] a real basis $\widehat{W}_q$ of ${\cal W}_q$ exists (see Remark \ref{remark:real-basis} on the existence of a real basis) and $\widehat{W}_q^\top \mathbb{J}_n \widehat{W}_q$ has full rank. 
\end{enumerate}
We highlight that the subspace is required to have dimension $2r$ to ensure that the the reduced system has $r$ degrees of freedom.
We also note that $\widehat{W}_q^\top \mathbb{J}_n \widehat{W}_q$ has full rank if and only if the test matrix \small $\left[\begin{array}{cc} 0_{r\times r} & \widehat{W}_q^\top \\ \mathbb{J}_n \widehat{W}_q & I_{n} \end{array}\right]$ \normalsize has full rank; see \cite[Eq. (3.4)]{MS74}.
From direct inspection, we can see that the test matrix would generally have full rank except for some $W_q$ with specific structures. In other words, $\{\sigma_1,\sigma_2,\ldots,\sigma_{2r}\}$ and $\{\mu_1,\mu_2,\ldots,\mu_{2r}\}$ can typically be chosen such that $\widehat{W}_q^\top \mathbb{J}_n \widehat{W}_q$ has full rank except for systems with specific forms of $A$ and $C$.

We then propose the projection matrices, $W_q$ and $V_q$, to be 
\begin{align}
 W_q &= \widehat{W}_q T^{\top} \label{eq:left-proj-quad} \\
 V_q &= \mathbb{J}_n W_q \left(W_q^\top \mathbb{J}_n W_q\right)^{-1} \label{eq:right-proj-quad}
\end{align}
where $T \in \mathbb{R}^{2r\times 2r}$ is a non-singular transformation matrix such that $T (\widehat{W}_q^\top \mathbb{J}_n \widehat{W}_q)T^\top = \mathbb{J}_r$. The existence of $T$ is guaranteed by Lemma \ref{lem:skewSymmetric} when Condition (iv) stated above holds.
Here, $W_q$ is also a real basis of ${\cal W}_q$ because $T$ is a non-singular matrix. Thus, the product $W_q^\top \mathbb{J}_n W_q$ is invertible because $\widehat{W}_q^\top \mathbb{J}_n \widehat{W}_q$ is assumed to have full rank. We highlight that, by choosing the above projection matrices, $W_q^\top V_q = I$ as required by Petrov-Galerkin projection approximation.
We now present a new result showing that, using the above projection matrices, our proposed reduced-order linear quantum stochastic system \eqref{eq:qsde-quad-reduced} is physically realizable while meeting the interpolation condition \eqref{eq:interpol-cond-left}.

\begin{theorem} \label{thm:pr-reduced-quad}
Given a physically realizable full-order linear quantum stochastic system written in quadrature form \eqref{eq:qsde-out-quad} with the transfer function $\Xi(s)$,
consider a reduced-order model of the form \eqref{eq:qsde-quad-reduced} with the projection matrices $W_q$ and $V_q$ given by \eqref{eq:left-proj-quad} and \eqref{eq:right-proj-quad}, respectively. Then the reduced-order system is also physically realizable and its transfer function $\Xi_r(s)$ interpolates $\Xi(s)$ in the sense of \eqref{eq:interpol-cond-left}.
\end{theorem}
\begin{proof}
First note that, because $W_q$ is a basis of the subspace ${\cal W}_q$ defined in \eqref{eq:Wsubspace-quad}, $W_q$ has full rank and \\ $\left(\mu_i^\dagger C(\sigma_i I - A)^{-1}\right)^\dagger \in {\rm range}(W_q)$. By Theorem \ref{thm:interpol-quad}(i), $\Xi_r(s)$ interpolates $\Xi(s)$ in the sense of \eqref{eq:interpol-cond-left}. 
 
Now we will show that the reduced-order system with the transfer function $\Xi_r(s)$ is physically realizable.
Let us first recall that $W_q = \widehat{W}_q T^{\top}$ where $T$ is an $2r\times 2r$ matrix such that 
$T (\widehat{W}_q^\top \mathbb{J}_n \widehat{W}_q)T^\top = \mathbb{J}_r$. Hence, we have that ${W}_q^\top \mathbb{J}_n {W}_q = \mathbb{J}_r$. It then follows from \eqref{eq:right-proj-quad} that $\mathbb{J}_n W_q = V_q \mathbb{J}_r$ and $W_q^\top \mathbb{J}_n = \mathbb{J}_r V_q^\top$.
Using these identities, pre- and post-multiplying the first physical realizability condition of the full-order system \eqref{eq:pr-1} by $W_q^\top$ and $W_q$, respectively, gives us
\begin{align*}
 W_q^\top A \mathbb{J}_n W_q + W_q^\top\mathbb{J}_n A^\top W_q + W_q^\top B\mathbb{J}_mB^\top W_q &= 0 \nonumber \\
 \iff A_r \mathbb{J}_r + \mathbb{J}_r A_r^\top + B_r\mathbb{J}_mB_r^\top &= 0. \label{eq:pr-reduced-quad-1}
\end{align*}

Similarly, pre-multiplying the second physical realizability condition of the full-order system \eqref{eq:pr-2} by $W_q^\top$ gives us that $\mathbb{J}_r C_r^{\top} +  B_r\mathbb{J}_mD^{\top} = 0$.
The theorem statement then follows because $D_r = D$.
\end{proof}

From the above theorem, we stress that the proposed projection matrices \eqref{eq:left-proj-quad}-\eqref{eq:right-proj-quad} lead to a physically realizable approximate system. Our proposed tangential interpolatory projection method can be applied to a large class of linear quantum stochastic systems compared to existing methods proposed in \cite{Nurd14,Pet11}, as will be demonstrated in examples. Moreover, the tangential interpolatory projection framework provides an appropriate reduction scheme when we are only interested in specific input-output responses of the system at a particular range of frequencies, 

Recall, for a subspace ${\cal H}$ of a vector space, that we write $P_{\cal H}$ to denote an orthogonal projection onto ${\cal H}$.
Let ${\cal X} \subset \mathbb{C}^{2n}$ and ${\cal Y} \subset \mathbb{C}^{2n}$.
We use $\phi\left( {\cal X},{\cal Y} \right)$ to denote the largest principal angle between the subspaces ${\cal X}$ and  ${\cal Y}$, which is defined as \cite{Szyld06}
\begin{equation}
 \cos \left( \phi\left( {\cal X},{\cal Y} \right) \right) \triangleq \inf_{\stackrel{x \in \mathcal{X},y \in \mathcal{Y}}{||x||=1,||y||=1}} \left| \langle x,y \rangle \right|.
 \label{eq:def-subspace-angle}
\end{equation}
Note that $\cos \left( \phi\left( {\cal X},{\cal Y} \right) \right) = \sqrt{1 - \|P_{\cal X} - P_{\cal Y}\|^2}$ \cite{Szyld06}.
We now present an error bound for the left tangential interpolatory projection method.

\begin{proposition} \label{prop:error-bound-left}
Consider a reduced-order model constructed using the subspace ${\cal W}_q$ defined in \eqref{eq:Wsubspace-quad}, and the full-rank projection matrices $W_q$ and $V_q$ as given in \eqref{eq:left-proj-quad} and \eqref{eq:right-proj-quad}, respectively. 
Let ${\cal U}_W(s) = \ker\{W_q^\top (sI-A)\}$ and ${\cal U}_V(s) = \ker\{V_q^\top (sI-A)^\dagger\}$.
For any $s \in \mathbb{C}$ that is not an eigenvalue of either $A$ or $A_r$, we have that
\begin{align}
\left\| \Xi(s) - \Xi_r(s) \right\| 
&= \left\| C(s  I-A)^{-1}  \left(I-Q(s)\right)  B \right\|,  \label{eq:left-err} \\
&= \left\| C\left( I - R(s)  \right) (sI-A)^{-1}  B\right\|, \label{eq:right-err}
\end{align} 
where $Q(s) = (sI-A) V_q (sI- A_r)^{-1} W_q^\top$ and $R(s) = V_q (sI-A_r)^{-1} W_q^\top(sI-A)$ are oblique projection operators (i.e., $Q(s)$ and $R(s)$ are idempotent), 
and when $A$ and $A_r$ have no poles on the right half plane and the imaginary axis, we have the following $\mathcal{H}_{\infty}$ bounds:
\begin{align}
 \left\| \Xi - \Xi_r \right\|_{\infty}
 &\leq \mathop{\sup}_{\omega \in \mathbb{R}} \left( \left({1 - \| P_{{\cal W}_q^\perp} - P_{{\cal U}_V(\imath\omega)} \|^2  }\right)^{-1/2} \right. \nonumber \\
 &\quad \left. \times \left\| C(\imath\omega I-A)^{-1} P_{{\cal W}_q^\perp} \right\|  \left\| P_{{\cal U}_V(\imath \omega)} B \right\|\right), \label{eq:left-err-bound} \\
 \left\| \Xi  - \Xi_r  \right\|_{\infty}  
 &\leq  \mathop{\sup}_{\omega \in \mathbb{R}} \left( \left({1 - \| P_{{\cal V}_q^\perp} - P_{{\cal U}_W(\imath\omega)} \|^2  }\right)^{-1/2} \right.  \nonumber \\
&\quad \left. \times \left\| C P_{{\cal U}_W(\imath \omega)} \right\| \left\| P_{{\cal V}_q^\perp} (\imath\omega I-A)^{-1} B  \right\|\right). \label{eq:right-err-bound} 
\end{align}
\end{proposition}
\begin{proof}
Let us first recall that the range and kernel of a projection operator $P$ are complementary in the sense that $\ker\{ I-P \} = {\rm range}\{ P \}$ \cite{Szyld06}. 
Under our assumption on $s \in \mathbb{C}$, 
$\ker\{Q(s)\} = {\rm range}\{I - Q(s)\} = {\cal W}_q^\perp$ because $(sI-A) V_q (sI-A_r)^{-1}$ is full rank.
Note that 
$\ker\{ I-Q(s) \}^\perp = {\rm range}( I - Q(s)^\dagger) \}$ \cite[Theorem 6.6]{Brezis2011}. From this identity, the fact that $W_q^\top (sI-A_r)^{-\dagger}$ is full rank, and the definition of $V_q$ \eqref{eq:right-proj-quad}, we have that $\ker\{ I-Q(s) \}^\perp = \ker\{ Q(s)^\dagger \} = \ker\{ V_q^\top (sI-A)^\dagger \} 
= {\cal U}_V(s)$.

From the above identities, we have that $I-Q(s) = P_{{\cal W}_q^\perp} \left(I-Q(s)\right) P_{{\cal U}(s)}$.
Thus, we have that
\begin{align*}
 &\left\| \Xi(s) - \Xi_r(s) \right\| \nonumber \\
 &= \left\| C(sI-A)^{-1} \left[ I - (sI-A)V_q (sI-A_r)^{-1} W_q^\top  \right] B\right\| \nonumber \\
 &= \left\| C(sI-A)^{-1} \left(I-Q(s)\right) B\right\| \nonumber \\
 &= \left\| C(sI-A)^{-1} P_{{\cal W}_q^\perp} \left(I-Q(s)\right) P_{{\cal U}_V(s)} B\right\|.
\end{align*}
This establishes \eqref{eq:left-err}. 
The result \eqref{eq:left-err-bound} then follows from the definition of the ${\cal H}_{\infty}$ norm and that $||I-Q(s)|| = \sec\left( {\rm range}\{I-Q(s)\}, \ker\{ I-Q(s) \}^\perp \right) = \left({1 - \| P_{{\cal W}_q^\perp} - P_{{\cal U}(\imath\omega)} \|^2  }\right)^{-1/2}$ \cite[Theorem 6.1]{Szyld06}.
Finally, \eqref{eq:right-err} and \eqref{eq:right-err-bound} follow from similar arguments to the above. This establishes the proposition statement.
\end{proof}

Note that for any $s \in \mathbb{C}$, 
$P_{{\cal W}_q^\perp}$, $P_{{\cal V}_q^\perp}$, $P_{{\cal U}_W(s)}$, and $P_{{\cal U}_V(s)}$ can be computed from \\ $\{\mu_1,\mu_2,\ldots,\mu_{2r}\}$, $\{\sigma_1,\sigma_2,\ldots,\sigma_{2r}\}$, $\mathbb{J}_n$, $A$, and $C$. 
Thus, the bound \eqref{eq:left-err-bound} can be obtained without computing $W_q$, $V_q$, and the reduced system model.

Moreover, we note that the two ${\cal H}_\infty$ bounds, \eqref{eq:left-err-bound} and \eqref{eq:right-err-bound}, are established on the basis of the operator norms of two different oblique projection operators $Q(s)$ and $R(s)$. 
Since they are oblique projection operators, their operator norms may be different and large (unlike orthogonal projections). Thus, the two upper bounds can be different and one might be tighter than the other, depending on $\{\mu_1,\mu_2,\ldots,\mu_{2r}\}$, $\{\sigma_1,\sigma_2,\ldots,\sigma_{2r}\}$, $\mathbb{J}_n$, $A$, and $C$.

\begin{remark} \label{remark:real-basis}
As we require $W_q$ and $V_q$ to be real-valued matrices, the interpolation points cannot be any arbitrary complex numbers 
and the tangent directions cannot be any arbitrary complex vectors. In fact, a real basis of the subspace ${\cal W}_q$ exists when the interpolation points, $\sigma_1,\ldots,\sigma_{2r}$, and their corresponding tangent directions, $\mu_1,\ldots,\mu_{2r}$, are real or occur in conjugate pairs; see also \cite[Remark 4]{GPBV12}.
\end{remark}

\subsubsection{Right Tangential Interpolatory Projection}

The right tangential interpolatory projection follows similar construction to the left tangential interpolatory projection presented previously. It involves a different subspace which is defined as
\begin{align*}
 {\cal V}_q 
 &\triangleq {\rm span}\left\{ (\sigma_iI-A)^{-1}B\nu_i \right\}_{i = 1,2,\ldots,2r}
\end{align*}
where interpolation points $\sigma_1,\sigma_2,\ldots,\sigma_{2r}$ and the right tangent directions $\nu_1,\nu_2,\ldots,\nu_{2r}$ are chosen such that: 
(i) $(\sigma_i I - A)$ is invertible for all $i = 1,2,\ldots,2r$, 
(ii) $\nu_i \neq 0$ for all $i = 1,2,\ldots,2r$,
(iii) the subspace ${\cal V}_q$ has the dimension of $2r$, and
(iv) a real basis $\widehat{V}_q$ of ${\cal V}_q$ exists and $\widehat{V}_q^\top \mathbb{J}_n \widehat{V}_q$ has full rank. 
We then propose  $V_q$ and $W_q$ to be 
\begin{align}
 V_q &= \widehat{V}_q T^{\top} \label{eq:right-proj-quad2} \\
 W_q &= \mathbb{J}_n V_q \left(V_q^\top \mathbb{J}_n V_q\right)^{-1} \label{eq:left-proj-quad2}
\end{align}
where $T \in \mathbb{R}^{2r\times 2r}$ is a non-singular transformation matrix such that $T \mathbb{J}_r T^\top = (\widehat{V}_q^\top \mathbb{J}_n^\top \widehat{V}_q)^{-1}$ and $V_q$ is a real basis of ${\cal V}_q$. 

\begin{theorem} \label{thm:pr-reduced-quad2}
Given a physically realizable full-order linear quantum stochastic system written in quadrature form \eqref{eq:qsde-out-quad} with the transfer function $\Xi(s)$,
consider a reduced-order model of the form \eqref{eq:qsde-quad-reduced} with the projection matrices $W_q$ and $V_q$ given by \eqref{eq:right-proj-quad2} and \eqref{eq:left-proj-quad2}, respectively. Then the reduced-order system is also physically realizable and its transfer function $\Xi_r(s)$ interpolates $\Xi(s)$ in the sense of \eqref{eq:interpol-cond-right}.
\end{theorem}
\begin{proof}
This proof follows similar argument to the proof of Theorem \ref{thm:pr-reduced-quad} using the result of Theorem \ref{thm:interpol-quad}(ii).
\end{proof}

\subsection{Selection of interpolation points and tangent directions} \label{subsec:interpol-selection}

In this subsection, we  propose a heuristic to select interpolation points and left tangent directions for our reduced-order linear quantum stochastic system. Note that we will only consider the left tangential interpolatory projection approach. However, similar approaches can be undertaken for the right tangential interpolatory projection approach. 

For the left tangential interpolatory projection, let us assume that the output fields $y(t)$ can be re-ordered from the most important pair (of momentum and position operators) to the least through a permutation matrix $\Pi_y$. We now present an approach to heuristically choose left tangent directions so that the more important output fields are matched at a larger (or at least equal) number of frequencies. Recall that $e_i = [0,0,\ldots,0,1,0,\ldots,0]^\top \in \mathbb{C}^{2\ell}$ denotes an indicator vector with one in the $i$-th position and zero elsewhere. Also let $E_k = (e_1,e_1,e_2,e_2,\ldots,e_{k},e_{k})$ for any $k \in \mathbb{N}$.
We propose the tangent directions to be
\begin{align}
 &(\mu_1,\mu_2,\ldots,\mu_{2r})   \nonumber \\
 &= \Pi_y^\top \times \left\{ 
 \begin{array}{cl}
   E_r, & \textnormal{if } r \leq \ell, \\
  (E_\ell,\ldots,E_\ell,e_1,e_1,e_2,e_2,\ldots), & \textnormal{otherwise}.
 \end{array} \right.
 \label{eq:chosen-direction}
\end{align}
Here, the tangent directions are chosen in conjugate pairs to ensure that a real basis of ${\cal W}_q$ exists.

It now remains to choose the interpolation points. In this approach, we are interested in matching frequency responses of the full-order system. Therefore, we propose that the set of interpolation points be along the imaginary axis of the form
\begin{equation}
 (\sigma_1,\sigma_2,\ldots,\sigma_{2r}) = (\imath \omega^c_1, -\imath \omega^c_1,\imath \omega^c_2, -\imath \omega^c_2,\ldots,\imath \omega^c_r, -\imath \omega^c_r)
 \label{eq:chosen-points}
\end{equation} 
where $\omega^c_1, \ldots,\omega^c_r \geq 0$. Again, the interpolation points are chosen in conjugate pairs to ensure that a real basis of ${\cal W}_q$ exists. 
For the chosen tangent directions and any $\Omega = (\omega_1,\omega_2,\ldots,\omega_{r}) \in \mathbb{R}^{r}$ (such that the conditions under \eqref{eq:Wsubspace-quad} hold), 
let $W(\Omega)$ be an orthonormal basis of the subspace 
$${\rm span}\left\{ (\imath\omega_j I \scalebox{0.75}[1]{\( - \)} A^\dagger)^{-1}C^\dagger \mu_k, (-\imath\omega_j I \scalebox{0.75}[1]{\( - \)} A^\dagger)^{-1}C^\dagger\mu_{k+1} \right\}_{j=1,2,\ldots,r}$$ where $k = 2j-1$. Let $V(\Omega) = \mathbb{J}_n W(\Omega) (W(\Omega)^\dagger\mathbb{J}_n W(\Omega))^{-1}$ and \\
$\tilde{Q}(s,\Omega) = (sI-A) V(\Omega) (sI- W(\Omega)^\dagger A V(\Omega) )^{-1} W(\Omega)^\dagger$.

\subsubsection{$\mathcal{H}_\infty$ based selection of interpolation points}

Since we are interested in matching frequency responses of the full-order system, it is natural to attempt to minimize the error between $\Xi(s) - \Xi_r(s)$ when $s$ is purely imaginary (i.e., minimize error across all frequencies).
Thus, we define a cost function on the basis of the ${\cal H}_\infty$ norm based on the error formulae \eqref{eq:left-err} as
\begin{align*}
 {\cal J}_\infty(\Omega) 
 &= \sup_{\omega \in \mathbb{R}}
       \left\| C(\imath\omega I-A)^{-1} \left(I-\tilde{Q}(\imath\omega,\Omega)\right) B\right\|.
\end{align*}
The vector of interpolation points $\Omega^c = (\omega^c_1,\omega^c_2,\ldots,\omega^c_r)$ is then chosen such that it is  a local   minimizer of ${\cal J}_{\infty}(\Omega)$.

\subsubsection{$\mathcal{H}_2$ based selection of interpolation points}

For some quantum systems, it is possible that the cost ${\cal J}_{\infty}(\Omega)$ is the same for all $\Omega \in \mathbb{R}^r$; see e.g., Example \ref{ex:passive2}.  In such cases, it may be more appropriate to consider an optimization problem based on ${\cal H}_2$ error.
Let $E(\omega,\Omega) = C(\imath\omega I-A)^{-1} \left(I-\tilde{Q}(\imath\omega,\Omega)\right)B$
Let us define a cost function on the basis of the ${\cal H}_2$ norm based on the error formulae \eqref{eq:left-err} as
\begin{align*}
 {\cal J}_2(\Omega) = \int_{-\infty}^{\infty} {\rm trace}\left( E(\omega,\Omega)^\dagger E(\omega,\Omega) \right) d\omega.
\end{align*}
The vector of interpolation points $\Omega^c = (\omega^c_1,\omega^c_2,\ldots,\omega^c_r)$ is then chosen such that it is a local minimizer of ${\cal J}_{2}(\Omega)$.

\subsection{Illustrative examples (active systems)}

In this subsection, we present some practical examples illustrating the application of our proposed tangential interpolatory projection approach on active linear quantum stochastic systems (having some squeezing). Particularly, the systems considered in these examples cannot be truncated through the existing quasi-balanced truncation method \cite{Nurd14} because the controllability and observability Gramians of the full-order systems do not satisfy the commutation condition for the existence of a symplectic balancing transformation matrix. 

\begin{example} \label{ex:BAE}
Consider an optomechanical system proposed in \cite{WC13} for back-action evading measurement of position, consisting of an optical cavity with two movable mirrors, conceptually illustrated in Fig.~\ref{fig:BAE1-model}. For an introduction to optomechanical systems we refer the reader to \cite{CM14,WM11}, while for studies of control on such systems  see, e.g., \cite{HM12,JNSJ15}. The cavity is pumped by a strong coherent laser and each mirror is subjected to radiation pressure and thermal noise. This system has three degrees of freedom comprising an oscillator inside the cavity, described by the quadratures $(q_1,p_1)$, and two mechanical oscillators from the motion of the two mirrors described by the quadratures $(q_2,p_2,q_3,p_3)$.
The dynamics of this optomechanical system can be linearized about the steady-state value of the mean amplitude of the cavity mode, see, e.g., \cite{CM14,WM11} for details. The linear approximation can be described in the quadrature form \eqref{eq:qsde-out-quad} with $x = (q_1,p_1,q_2,p_2,q_3,p_3)^\top$ and the following system matrices \cite{WC13}:
\begin{align*}
 A &= \left[\begin{array}{cccccc}
       -\frac{\kappa}{2} & 0 & 0 & 0 & 0 & 0 \\
       0 & -\frac{\kappa}{2} & -\Gamma & 0 & 0 & 0 \\
       0 & 0 & -\frac{\gamma}{2} & 0 & 0 & \Omega_b \\
       -\Gamma & 0 & 0 & -\frac{\gamma}{2} & -\Omega_b & 0 \\
       0 & 0 & 0 & \Omega_b & -\frac{\gamma}{2} & 0 \\
       0 & 0 & -\Omega_b & 0 & 0 & -\frac{\gamma}{2}
      \end{array}\right],
\end{align*}
$B = {\rm diag}(\sqrt{\kappa}I_2, \sqrt{\gamma}I_4)$, $C = \left[\sqrt{\kappa}I_2 \quad 0_{2\times 4}\right]$, and $D = \left[-I_2 \quad 0_{2\times 4}\right]$, where $\kappa > 0$ is the cavity decay rate, $\gamma > 0$ is the damping rate of the two mechanical oscillators, $\Gamma > 0$ is the optomechanical coupling rate (due the coupling between the mirror degrees of freedom and the cavity mode via radiations pressure), and $\Omega_b$ is the system half-bandwidth. 
Here, inputs 1-2 are the two quadratures of the laser field while inputs 3-6 describe the thermal fluctuations acting on the mirrors. Note that this system is active as the coupling $\Gamma$ leads to a squeezing Hamiltonian.

\begin{figure}[!t]
\centering
\includegraphics[width=0.45\textwidth]{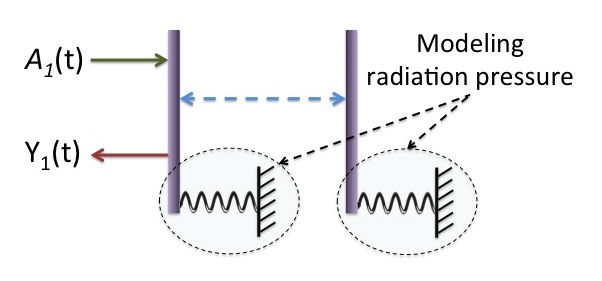}
\caption{Example \ref{ex:BAE}: Conceptual model of the optomechanical system. }
\label{fig:BAE1-model}
\end{figure}

Let $\kappa = 2\times 10^5$, $\gamma = 100$, $\Gamma = 7.0711 \times 10^4$, and $\Omega = 10^4$. 
The poles of the system are then at $-50 \pm 10^4\imath$ and $-10^5$.
This choice of parameter values allows us to compare our model reduction method with the singular perturbation approximation presented in \cite{VP11}.

Consider an approximation with $r=2$. We are interested in matching frequency responses from the thermal fluctuations on the mechanical oscillators to the quadratures of the output field $Y_1(t)$. We apply the right tangential interpolatory projection model reduction approach to provide different weights on the inputs. As proposed in \eqref{eq:chosen-direction}, we choose the tangent directions to be 
$(\nu_1,\nu_2,\nu_3,\nu_4) = (e_5,e_5,e_6,e_6)$. The corresponding interpolation points of the form \eqref{eq:chosen-points} are obtained by finding a local minimizer of the ${\cal H}_\infty$ based optimization problem. For computational simplicity, we set $\omega^c_1 = \omega^c_2 = \omega^c$. With this assumption, we find $\omega^c = 1.05\times 10^4$ as a local minimizer. We note that this $\omega^c$ leads to a stable reduced system whose poles are at $-50 \pm 10^4\imath$.

\begin{figure}[!t]
\centering
\includegraphics[width=0.6\textwidth]{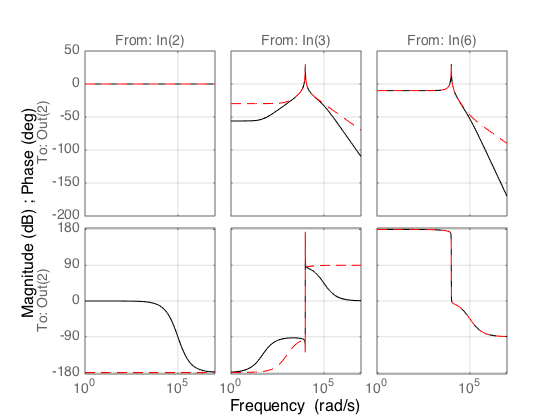}
\caption{Example \ref{ex:BAE}: The frequency responses (Bode plot) of the full-order linear quantum stochastic system (black solid line) and the tangential interpolatory projection approximation (red dashed line). The center frequency (where the peaks appear) is the average resonant frequencies of the two mechanical oscillators, i.e., $(\omega_{m1}+\omega_{m2})/2$.}
\label{fig:BAE1-1}
\end{figure}
\begin{figure}[!t]
\centering
\includegraphics[width=0.6\textwidth]{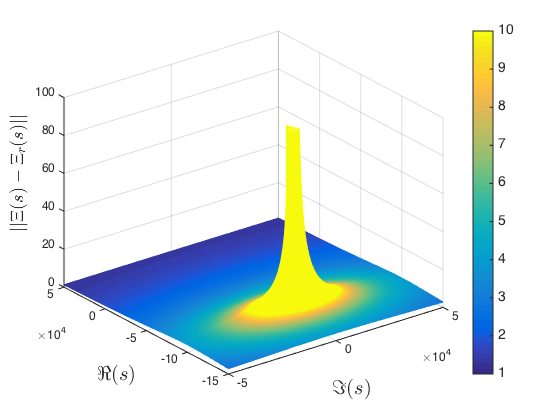}
\caption{Example \ref{ex:BAE}: The approximation error $\|(\Xi(s)-\Xi_r(s))\|$.}
\label{fig:BAE1-err}
\end{figure}
\begin{figure}[!t]
\centering
\includegraphics[width=0.5\textwidth]{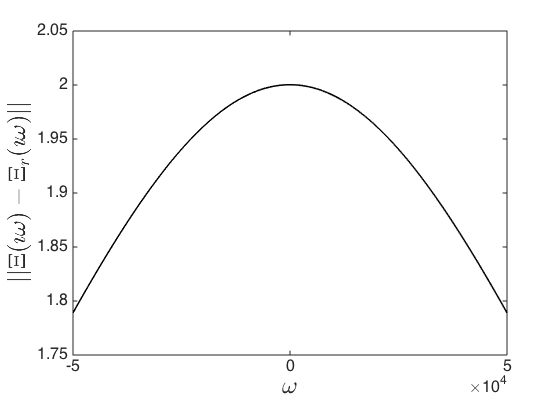}
\caption{Example \ref{ex:BAE}: The approximation error at each purely imaginary point (or frequency) $\|(\Xi(\imath\omega)-\Xi_r(\imath\omega))\|$.}
\label{fig:BAE1-err2}
\end{figure}

Fig. \ref{fig:BAE1-1} illustrates the frequency responses of the original and our reduced systems at output 2, using black solid and red dashed lines, respectively. Note that the responses from inputs 1, 4, and 5 are omitted as the transfer functions from these inputs are zero due to the dispersive coupling. Similarly, the responses at output 1 are also omitted because, for both systems, the transfer functions from inputs 2-6 are zero and the ones from input 1 are the same as those at output 2 from input 2 (shown in the figure).
From the figure, we see that the responses of the two systems are well matched over the narrow bandwidths ($2\Omega_b$) of the forces (inputs 3 and 6) acting on the mechanical oscillators. The magnitude responses from the pump beam (input 2) of the two systems are the same, whilst the phase response of our reduced system matches that of the original system at high frequencies. 
The error $\|(\Xi(s)-\Xi_r(s))\|$ for each $s$ is shown in Fig. \ref{fig:BAE1-err} and \ref{fig:BAE1-err2}. Note that the tangential error is large when $s$ is close to the poles of the systems. The ${\cal H}_\infty$ error of the reduced system is $2.00$. The ${\cal H}_\infty$ bounds \eqref{eq:left-err-bound} and \eqref{eq:right-err-bound} are $2.45$ and $3.96\times 10^3$, respectively.
Thus, we see that the bound \eqref{eq:left-err-bound} is tight, while the bound \eqref{eq:right-err-bound} is too conservative for this example.

For comparison, singular perturbation approximation \cite{VP11} is applied to eliminate the cavity oscillator $(q_1,p_1)$ and the ${\cal H}_\infty$ error of the singular perturbation approximation is $4.45$. Hence, our model reduction method provides a better reduced model than the singular perturbation approximation (in terms of ${\cal H}_\infty$ error) for this optomechanical example with the given set of parameter values.

\end{example}

\begin{example} \label{ex:control}
Consider a quantum control system originally considered in \cite[Section IV]{SP12}, comprising of a cascade connection of an optical parametric oscillator (OPO), two optical cavities, and a stabilizing linear quantum controller. Let $u_q(t)$ and $u_p(t)$ denote the momentum and position quadratures of the output of the stabilizing controller, respectively. The dynamics of the overall control system is described by
\begin{align*}
 dx(t) &= Ax(t) dt + B_u du(t) + Bdw(t) \nonumber \\
 dy(t) &= Cx(t) dt + D dw(t)
\end{align*}
where $du(t) = [u_q(t),u_p(t)]^\top dt + dv(t)$ for some vacuum quantum noise vector $v(t)$ different from $w(t)$ originating from the controller output field, and
\small
\begin{align*}
 A &= \left[\begin{array}{ccc} -0.5006I_2 & -0.0374I_2 & -0.0410I_2 \\ 0 & A_{22} & -1.0954I_2 \\ 0 & 0 & -0.6I_2  \end{array}\right], \\
 A_{22} &= \left[\begin{array}{cc} -1 & -1.05 \\ -1.05 & -1 \end{array}\right], \\
 B_u &= \left[\begin{array}{ccc} -0.0374I_2 & -I_2 & -1.0954I_2 \end{array}\right]^\top, \\
 C &= \left[\begin{array}{cc} I_2 & 0_{2\times 4} \end{array}\right]^\top, \quad
 D = \left[\begin{array}{cc} I_2 & 0_{2\times 2} \end{array}\right]^\top.
\end{align*}
\normalsize
Here, the plant is unstable (the eigenvalues of $A$ are $0.05$, $-0.05$, $-0.5007$, and $-2.05$).

To stabilize the plant, we first design a classical (non-quantum) output feedback controller using the pole placement method. We select the closed-loop poles to be at $-0.2$, $-0.3$, $-0.5$, $-0.6$, $-0.9$, $-1.5$, $-1.8$. The controller dynamics is then described by
\begin{align*}
 dz(t) &= A_c z(t) dt + B_c dy(t) \nonumber \\
 du(t) &= C_c z(t) dt
\end{align*}
where
{\small \arraycolsep=2pt\def\arraystretch{1}
\begin{align*}
 A_c &\hspace{-0.3em}=\hspace{-0.3em} \left[\hspace{-0.3em}\begin{array}{rrrrrr}
    -1.5500 & -0.0001 &  -0.0016 &  0.0000 &  -0.0139 &   0.0000 \\
    -0.0052  & -2.4500 &   0.0000  & -0.0315  &  0.0000 &  -0.0447 \\
    -10.3270 &  -0.0022 &  -1.0920 &  -0.0002 &  -0.3718 &   0.0004 \\
     0.2787 &  39.3723  &  0.0003 &   0.2090 &   0.0001  & -1.1943 \\
    17.5312 &   0.0014 &   1.0494 &  -0.0002  &  0.1927  &  0.0005 \\
    -0.0207  &  0.0007 &   0.0003 &   0.1741 &   0.0002 &  -0.7083
 \end{array}\hspace{-0.3em}\right]\hspace{-0.3em}, \\
 B_c &\hspace{-0.3em}=\hspace{-0.3em} \left[\hspace{-0.3em}\begin{array}{rr}
     1.0493 &   0.0001 \\
         0.0052  &  1.9493 \\
        10.3276  &  0.0022 \\
        -0.2787 & -39.3717 \\
       -17.5305 &  -0.0014 \\
         0.0207  &  0.0000
  \end{array}\hspace{-0.2em}\right]\hspace{-0.3em}, \quad
 C_c \hspace{-0.2em}=\hspace{-0.3em} \left[\hspace{-0.2em}\begin{array}{rr}
     -0.0006   & 0.0000 \\
        -0.0000 &   -0.0007 \\
        -0.9580 &  -0.0003 \\
         0.0002  & -0.1590 \\
        -0.7236 &  -0.0001 \\
        -0.0004  &  0.0989
  \end{array}\hspace{-0.2em}\right]^\top\hspace{-0.3em}.
\end{align*}}
This type of controller can be made physically realizable as a linear quantum system with six degrees of freedom by introducing seven additional inputs $v(t)$ to the controller \cite[Lemma 5.6]{JNP06}. That is, there exists a matrix $B_v$ such that the system described by 
\begin{align*}
 dz(t) &= A_c z(t) dt + [B_v \quad B_c] [dv(t)^\top \quad dy(t)^\top]^\top \nonumber \\
 du(t) &= C_c z(t) dt + [I_2 \quad 0] [dv(t)^\top \quad dy(t)^\top]^\top 
\end{align*}
is physically realizable (as a quantum system). Using the formula given in \cite[Lemma 5.6]{JNP06}, $B_v \in \mathbb{R}^{6\times 14}$ is given by
{\small \arraycolsep=2pt\def\arraystretch{1}
\begin{align*}
 \left[\hspace{-0.3em}\begin{array}{rrrrrr}
     0.0007  &  0.0000  &  0.1590 &  -0.0003 &  -0.0989 &  -0.0001 \\ 
       -0.0000  &  0.0006 &   0.0002  &  0.9580 &  -0.0004  &  0.7236 \\
       -0.0328 &   0.0000 &   0.1957  & -0.0012 &  -0.3337 &   0.0004 \\
       -0.0000 & -29.6894  & -0.0002 &  -0.1134  &  0.0003  &  0.0001 \\
             0  & -0.0328 &  -0.0000  & -0.0387 &   0.0000 &  -0.0002 \\
       29.6921  &  0.0000 &  -0.0056  &  0.0000  &  0.0096 &  -0.0000 \\
       -0.0387  &  0.0012 &  -8.1505  & -0.0000 &  13.8054 &  -0.0163 \\
       -0.0000 &  -0.1134 &  -0.0001 & -24.9918 &  -0.0001 &  -0.0027 \\
        0.0000  &  0.1957 &  -0.0000  & -8.1505  &  0.0005 &  -0.0006 \\
       -0.0056 &   0.0002 &  28.4766 &   0.0001 &   2.0589 &  -0.0024 \\
       -0.0002 &  -0.0004 &  -0.0006  &  0.0163 &  -0.0112  &       0 \\
        0.0000  &  0.0001  &  0.0024  & -0.0027 &  -0.0041 & -29.6921 \\
       -0.0000 &  -0.3337 &  -0.0005  & 13.8054 &   0.0000 &  -0.0112 \\
        0.0096 &  -0.0003  &  2.0589  &  0.0001 &  26.2046  &  0.0041 \\
 \end{array}\hspace{-0.2em}\right]\hspace{-0.3em}.
\end{align*}}
Note that $B_v$ does not affect the location of the poles of the closed-loop system; see \cite[Eq. (24)]{JNP06}.
Also, note that the controller is stable with the poles at $-0.3837$, $-0.4597\pm -0.6664\imath$, $-0.7930$, and $-1.2726$, $-2.0299$.

In this example, we are interested in approximating the controller described by \\
$(A_c,[B_v \ B_c],C_c,[I_2 \ 0])$ with $r = 2$. To place more emphasis on the input $dy(t)$ of the controller (output of the plant), we consider the right tangential interpolatory projection method.  We choose $(\nu_1,\nu_2,\nu_3,\nu_4) = (e_{15},e_{15},e_{16},e_{16})$. 
The corresponding interpolation points of the form \eqref{eq:chosen-points} are obtained from a local minimizer of the ${\cal H}_{\infty}$ optimization problem described in the previous subsection. Again, we simplify the problem by assuming that $\omega^c_1 = \omega^c_2 = \omega^c$. With this assumption, we have that $\omega^c = 0.2900$. This $\omega^c$ leads to a stable reduced controller with the poles at $-0.2576\pm 1.4795\imath$, $-0.5391$, and $-1.4958$.

\begin{figure}[!t]
\centering
\includegraphics[width=0.6\textwidth]{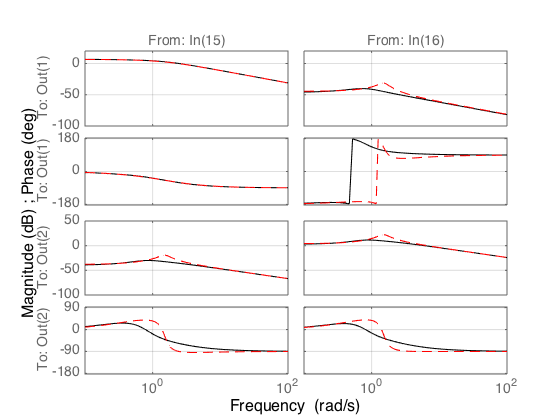}
\caption{Example \ref{ex:control}: The frequency responses (Bode plot) of the full-order controller (black solid line) and the tangential interpolatory projection approximation (red dashed line) from inputs 15 and 16 (i.e., output of the plant $y(t)$).}
\label{fig:control-ex}
\end{figure}

Fig. \ref{fig:control-ex} illustrates the frequency responses of the original controller and the tangential interpolatory projection approximation from the input $y(t)$. We see that the responses at output 1 from input 15 is closely matched across all frequencies. The other pairs of input-output responses from the quadratures in $y(t)$ to quadratures of the controller output $u(t)$ are also quite well matched except around the frequency of 0-2 rad/s. The closed-loop system with the approximate controller has poles at $-0.1265\pm 0.1404\imath$, $-0.3272$, $-0.3654\pm 1.5331\imath$, $-0.5821$, $-0.6220$, $-0.7143$, and $-1.7610 \pm 0.1907\imath$. That is, the approximate reduced quantum controller with four degrees of freedom also stabilizes the closed-loop system.
\end{example}

\section{Model Reduction for Completely Passive Linear Quantum Stochastic Systems}
\label{sec:passive}

Completely passive linear quantum stochastic systems are those which can be physically implemented using only passive components (that do not require external sources of quanta or energy); for example in quantum optics, these components are optical cavities, beam splitters, and phase shifters.
In \cite{TN15}, an asymptotically stable linear quantum stochastic system, in the quadrature form \eqref{eq:qsde-out-quad}, is shown to be completely passive if and only if its controllability Gramian $P$ is identity, i.e., $P = I$. The proposed choices of projection matrices $W_q$ and $V_q$ in the previous section does not guarantee that the controllability Gramian $P_r$ of the reduced system in the quadrature form \eqref{eq:qsde-quad-reduced} is $P_r = I$ (when the reduced system is asymptotically stable). 
For this reason, we will now present a model reduction method for completely passive linear quantum stochastic systems, which guarantees both passivity and physical realizability properties of the reduced system, by appealing to the system dynamics described by annihilation operators.

We begin by noting that each pair of position and momentum operators $(q_j,p_j)$ can be associated with a mode (or annihilation operator) as $a_j = (q_j+\imath p_j)/2$, which serves as a quantum analogue of the amplitude of a harmonic oscillator.
Let $a = (a_1,a_2,\ldots,a_n)^\top$ denote a vector of annihilation operators of a $n$ degree of freedom quantum harmonic oscillators.
The dynamics of a linear quantum stochastic system can be completely represented in the following annihilation-operator form if and only if the system is completely passive \cite{GGY08,GZ15}:
\begin{align}
 da(t) &= F a(t) dt + G d{\cal A}(t); \quad a(0) = a \nonumber \\
 dY(t) &= H a(t) dt + K d{\cal A}(t), \label{eq:qsde-anni}
\end{align}
where $F \in \mathbb{C}^{n\times n}$, $G \in \mathbb{C}^{m\times n}$, $H \in \mathbb{C}^{\ell \times n}$, and $D \in \mathbb{C}^{\ell \times m}$.
Here, $Y(t)$ and ${\cal A}(t)$ are as defined in Section \ref{sec:prelim}.

We will now present the physical realizability conditions of a linear quantum stochastic system in annihilation-operator form \eqref{eq:qsde-anni}. Similar to the earlier works  \cite{Nurd14,JNP06}, when the number of outputs is less than inputs ($\ell < m$), we say that a completely passive linear quantum stochastic system is physically realizable if and only if there exists $\hat{H} \in \mathbb{C}^{(m-\ell)\times n}$ and $\hat{K} \in \mathbb{C}^{(m-\ell)\times m}$ such that the following augmented system (with $\ell = m$) is physically realizable:
\begin{align}
 da(t) &= F a(t) dt + G d{\cal A}(t); \quad a(0) = a \nonumber \\
 dY_a(t) &= \left[\begin{array}{cc} H \\ \hat{H} \end{array}\right] a(t) dt + \left[\begin{array}{cc} K \\ \hat{K} \end{array}\right] d{\cal A}(t) 
\end{align}
where $Y_a(t)$ is a $m$-dimensional output vector.

\begin{theorem}
A linear quantum stochastic system in annihilation-operator form \eqref{eq:qsde-anni} (with $\ell \leq m$) is physically realizable if and only if 
\begin{align}
 F + F^\dagger + GG^\dagger = 0 \label{eq:pr-anni-1} \\
 H^\dagger + G K^\dagger = 0 \label{eq:pr-anni-2} 
\end{align}
and there exists $\hat{K} \in \mathbb{C}^{(m-\ell)\times m}$ such that the matrix $\tilde{K} = (K^\top, \hat{K}^\top)^\top$ is a complex unitary matrix in the sense that
\begin{equation}
 \tilde{K}\tilde{K}^\dagger = \tilde{K}^\dagger \tilde{K} = I.  \label{eq:pr-anni-3}
\end{equation}
Note that when $\ell = m$, $\tilde{K} = K$.
\end{theorem}
\begin{proof}
When $\ell = m$, the necessity and sufficiency of \eqref{eq:pr-anni-1}-\eqref{eq:pr-anni-3} are shown in \cite{GGY08,GZ15}.
When $\ell < m$, the necessity and sufficiency of \eqref{eq:pr-anni-1} and \eqref{eq:pr-anni-3} follows immediately from the corresponding physical realizability conditions of the augmented system with $\ell = m$. 
Let $\tilde{H} = (H,\hat{H})^\top$ for some $\hat{H} \in \mathbb{C}^{(m-\ell)\times n}$. For necessity of \eqref{eq:pr-anni-2}, we note from the corresponding physical realizability condition of the augmented system that $\tilde{H}^\dagger = -\tilde{G}\tilde{K}^\dagger$. Post-multiplying both sides of this equation by $(I_\ell, 0)^\top$, we have that ${H}^\dagger = -{G}{K}^\dagger$. For sufficiency of \eqref{eq:pr-anni-2}, let $\hat{H} = -\hat{K}G^\dagger$ where $\hat{K}$ is the complex matrix that makes \eqref{eq:pr-anni-3} holds. From this $\hat{H}$ and \eqref{eq:pr-anni-2}, we have that $\tilde{H}^\dagger = -\tilde{G}\tilde{K}^\dagger$. This establishes the theorem statement.
\end{proof}

\subsection{Reduced-order completely passive linear quantum stochastic systems}

Similar to the previous section, we will construct a reduced-order model of a completely linear quantum stochastic system 
via Galerkin projection.
That is, given a full-order linear quantum stochastic system in annihilation form \eqref{eq:qsde-anni}, 
we seek a reduced-order linear quantum stochastic system of the form
\begin{align}
 da_r(t) &= F_r a_r(t) dt + G_r d{\cal A}(t); \quad a_r(0) = V^\dagger a \nonumber \\
 dY_r(t) &= H_r a_r(t) dt + K_r d{\cal A}(t)  \label{eq:qsde-anni-reduced}
\end{align}
where $F_r = V_a^\dagger F V_a$, $G_r = V_a^\dagger G$, $H_r = HV_a$, and $K_r =K$. Here, the projection matrix $V_a \in \mathbb{C}^{n\times r}$ is orthonormal in the sense that $V_a^\dagger V_a = I_r$.

We now propose our tangential interpolatory projection model reduction method. We will only present the left tangential interpolatory projection in this section, but the right tangential interpolatory projection can be similarly constructed involving a different subspace as demonstrated previously.
To achieve the left tangential interpolation condition \eqref{eq:interpol-cond-left}, consider a subspace
\begin{align}
 {\cal V}_a 
 &\triangleq {\rm span}\left\{ \left(\mu_i^\dagger H(\sigma_iI-F)^{-1}\right)^\dagger \right\}_{i = 1,2,\ldots,r}
 \label{eq:Vsubspace-anni}
\end{align}
where interpolation points $\sigma_1,\sigma_2,\ldots,\sigma_{r}$ and the left tangent directions $\mu_1,\mu_2,\ldots,\mu_{r}$ are chosen such that:
(i) $(\sigma_i I - F)$ is invertible for all $i = 1,2,\ldots,r$,
(ii) $\mu_i \neq 0$ for all $i = 1,2,\ldots,r$, and
(iii) the subspace ${\cal V}_a$ has the dimension  $r$.
Again, the condition on the dimension of ${\cal V}_a$ ensures that the reduced-order system has $r$ degree of freedom. The interpolation points and tangent directions can be chosen as discussed in Section \ref{subsec:interpol-selection}, except that the points and directions do not need to be in conjugate pairs as $V_a$ can be a complex matrix. The following result is straightforward to show, so we shall simply state it without proof.

\begin{theorem} \label{thm:pr-reduced-anni}
Given a full-order linear quantum stochastic system written in the annihilation-operator form \eqref{eq:qsde-anni} with the transfer function $\Xi(s)$,
consider a reduced-order model of the form \eqref{eq:qsde-anni-reduced} with the projection matrix $V_a$ which is an orthonormal basis of the subspace ${\cal V}_a$ defined in \eqref{eq:Vsubspace-anni}.
Then the reduced-order model is physically realizable and its transfer function $\Xi_r(s)$ interpolates $\Xi(s)$ in the sense of \eqref{eq:interpol-cond-left}.
\end{theorem}

The tangential interpolatory projection method provides an avenue to reduce the order of a linear quantum stochastic system whilst ensuring that the reduced system is both physical realizable (as shown in Theorem \ref{thm:pr-reduced-anni}) and completely passive (because the system dynamics are in annihilation-operator form). A key advantage of our method is that it is applicable to a larger class of completely passive linear quantum stochastic systems in comparison to some existing methods such as the quasi-balanced truncation method \cite{Nurd14} (which is applicable to those with unequal Hankel singular value of the product of the system's controllability and observability Gramians, meaning that the system has less inputs than it does outputs), and the eigenvalue truncation method \cite{Pet12} (which is applicable to systems with distinct eigenvalues). The following result is the analogue of Proposition \ref{prop:error-bound-left} for completely passive systems and is proved in an analogous way.

\begin{proposition}
Consider a reduced-order model of the form \eqref{eq:qsde-anni-reduced} with the projection matrix $V_a$ which is an orthonormal basis of the subspace ${\cal V}_a$ defined in \eqref{eq:Vsubspace-anni}.
Let ${\cal U}_1(s) = \ker\{ V_a^\dagger (sI-F)^\dagger \}$ and ${\cal U}_2(s) = \ker\{ V_a^\dagger (sI-F) \}$. 
For any $s \in \mathbb{C}$ that is not an eigenvalue of either $F$ or $F_r$, we have that
\begin{align}
 \left\| \Xi(s) - \Xi_r(s) \right\| 
 &= \left\| H(sI-F)^{-1}(I-Q(s)) G \right\|, \\
 &= \left\| H(I-R(s)) (sI-F)^{-1} G \right\|,
\end{align}
where $Q(s) = (sI - F) V_a (sI-F_r)^{-1} V_a^\dagger$ and $R(s) = V_a (sI-F_r)^{-1} V_a^\dagger(sI-F)$ are oblique projection operators (i.e., $Q(s)$ and $R(s)$ are idempotent),
and when $F$ and $F_r$ have no poles on the right half plane and imaginary axis, we have the following ${\cal H}_{\infty}$ bounds:
\begin{align}
 \left\| \Xi - \Xi_r \right\|_{\infty}
 &\leq \mathop{\sup}_{\omega \in \mathbb{R}} \left( \left({1 - \| P_{{\cal V}_a^\perp} - P_{{\cal U}_1(\imath\omega)} \|^2  }\right)^{-1/2} \right. \nonumber \\
 &\quad \left. \times \left\| H(\imath\omega I-F)^{-1} P_{{\cal V}_a^\perp} \right\|  \left\| P_{{\cal U}_1(\imath \omega)} G \right\|\right), \label{eq:passive-left-err-bound} \\
 \left\| \Xi  - \Xi_r  \right\|_{\infty}  
 &\leq  \mathop{\sup}_{\omega \in \mathbb{R}} \left( \left({1 - \| P_{{\cal V}_a^\perp} - P_{{\cal U}_2(\imath\omega)} \|^2  }\right)^{-1/2} \right.  \nonumber \\
&\quad \left. \times \left\| H P_{{\cal U}_2(\imath \omega)} \right\| \left\| P_{{\cal V}_a^\perp} (\imath\omega I-F)^{-1} G  \right\|\right). \label{eq:passive-right-err-bound} 
\end{align}
\end{proposition}

\subsection{Stability Property}

In this subsection, we will present some sufficient conditions that guarantee asymptotic stability of the reduced-order completely passive linear quantum stochastic system.

\begin{lemma} \label{lem:stab-passive}
Given a linear quantum stochastic system written in annihilation-operator form \eqref{eq:qsde-anni},
consider a reduced-order model of the form \eqref{eq:qsde-anni-reduced} with a projection matrix $V_a$ that is an orthonormal basis of the subspace ${\cal V}_a$ defined in \eqref{eq:Vsubspace-anni}. The reduced-order system $(F_r,G_r,H_r,K_r)$ is asymptotically stable and is a minimal (controllable and observable) realization if and only if 
\begin{equation}
 G^\dagger V_a z_r \neq 0, \label{eq:stab-iff}
\end{equation}
for each eigenvector $z_r$ of $F_r$.
Moreover, \eqref{eq:stab-iff} holds whenever one of the following sufficient conditions is satisfied:
\begin{enumerate}
 \item[(i)] ${\cal V}_a \subseteq {\ker\{G^\dagger \}}^\perp$.
 \item[(ii)] the interpolation points $\sigma_i \neq \imath \omega_r$ for all $i = 1,2,\ldots,r$ where $\imath \omega_r$ is an eigenvalue of the resulting reduced-order system matrix $F_r$. 
\end{enumerate}
\end{lemma}
\begin{proof}
First note from Theorem \eqref{thm:pr-reduced-anni} that the reduced system satisfies 
\begin{align*}
 F_r + F_r^\dagger + G_rG_r^\dagger = 0
\end{align*} 
Pre- and post-multiplying the above equation by $z_r^\dagger$ and $z_r$, respectively, we have that
\begin{align}
 2\Re\{\lambda_r \} \|z_r\|^2 + \left\|G_r^\dagger z_r\right\|^2 = 0, \label{eq:stab-proof1}
\end{align} 
where $\lambda_r$ is the eigenvalue of $F_r$ corresponding to $z_r$.
This implies that the system is stable (i.e., $\Re\{\lambda_r\} < 0$ for all eigenvalues $\lambda_r$) if and only if $\left\|G_r^\dagger z_r\right\| = \left\|G^\dagger V_a z_r\right\| \neq 0$ for all eigenvectors $z_r$. Minimality  of the reduced model then follows from \cite[Lemma 2]{GZ15}, where it is shown that  for completely passive linear quantum systems asymptotic stability is equivalent to minimality.

We now show the sufficiency of Conditions (i) and (ii).

{\em Sufficiency of Condition (i):} Since $V_a$ is an orthonormal basis of ${\cal V}_a$, we have that $V_a z_r \in {\cal V}_a$. 
Now because ${\cal V}_a \subseteq {\ker\{ B^\dagger \}}^\perp$, $G^\dagger V_a z_r \neq 0$. 

{\em Sufficiency of Condition (ii):} First note from \eqref{eq:stab-proof1} that the poles of the reduced system do not lie on the right-half plane.
We then prove the sufficiency of this condition by showing that when $G_r^\dagger z_r = 0$, the interpolation point is $\sigma_i = \imath \omega_r$ for some $i = 1,2,\ldots,r$. Note that this proof follows similar construction to the proof of \cite[Theorem 11]{GPBV12}.
Suppose that $\Xi_r(s)$ has a pole on the imaginary axis $\imath \omega_r$. 
Let $z_r$ be the eigenvector corresponding to the eigenvalue $\imath\omega_r$ of $F_r$, i.e., $F_r z_r = \imath \omega_r z_r$.
Also note, from the physical realizability of the reduce system (established in Theorem \ref{thm:pr-reduced-anni}), that $H_r = -K_r G_r^\dagger $. Therefore, from this identity and \eqref{eq:stab-proof1}, we have that $H_r z_r = 0$ whenever $G_r^\dagger z_r = 0$. 

Now let $v_i = (\mu_i^\dagger H(\sigma_iI - F)^{-1})^\dagger$. We also let $V = [v_1\;v_2\;\ldots\;v_r]$ and $U = [\mu_1\;\mu_2\;\ldots\;\mu_r]$. From the definitions of $v_i$, $V$, and $U$, we have that
\begin{equation}
 V^\dagger F - \Sigma V^\dagger + U^\dagger H = 0 \label{eq:stab-anni-3}
\end{equation}
where $\Sigma = {\rm diag}(\sigma_1,\sigma_2,\ldots,\sigma_r)$ is a diagonal matrix with the interpolation points on its diagonal. Note that $V$ is also a basis of ${\cal V}_a$.
Since the projection matrix $V_a$ used to construct the reduced system is an orthonormal basis of ${\cal V}_a$,
there exists a non-singular (full-rank) matrix $T_a \in \mathbb{C}^{r \times r}$ such that $V = V_a T_a$.
Substituting $V_a T_a$ for $V$ and post-multiplying \eqref{eq:stab-anni-3} by $V_a z_r$, we have that
\begin{align}
 T_a^\dagger V_a^\dagger F V_a z_r - \Sigma T_a^\dagger V_a^\dagger V_a z_r + U^\dagger H V_a z_r &= 0 \nonumber \\
 T_a^\dagger F_r z_r - \Sigma T_a^\dagger z_r + U^\dagger H_r z_r &= 0 \nonumber \\
 \imath \omega_r T_a^\dagger z_r - \Sigma T_a^\dagger z_r &= 0 \nonumber \\
 \left(\imath \omega_r I_r - \Sigma \right) T_a^\dagger z_r &= 0 \label{eq:stab-anni-4}
\end{align}
The 3rd step follows from the previously obtained result that $H_r z_r = 0$.
Since $T_a$ has full rank and $z_r$ is non-trivial (because it is an eigenvector), we have that $T_a^\dagger z_r \neq 0$.
Moreover, because $\Sigma = {\rm diag}(\sigma_1,\sigma_2,\ldots,\sigma_r)$, the result \eqref{eq:stab-anni-4} implies that $\sigma_i = \imath \omega_r$ for some $i = 1,2,\ldots,r$.
This establishes the sufficiency of Condition (ii) and hence the lemma statement.
\end{proof}

We stress that the Condition (ii) of Lemma \ref{lem:stab-passive} is satisfied in many situations because it is generally improbable to find a purely imaginary interpolation point such that the point matches the imaginary pole of the resulting reduced system by chance. Hence, our proposed model reduction method for completely passive systems will generically lead to an asymptotically stable reduced-order linear quantum stochastic system.

\subsection{Illustrative example (passive system)}

\begin{example} \label{ex:passive2}
Consider a cascade connection of five identical optical cavities with the decay rate $\gamma = 10^6$ and the resonant frequency $\omega_0$ ($\omega_0$ is much larger than $\gamma$ but the exact value is not important here). Each cavity consists of two mirrors labeled M1 and M2. The mirrors M1 and M2 of cavity 1 are driven by coherent light  fields ${\cal A}_1$ and ${\cal A}_2$, respectively. The amplitudes of these fields are modulated by a carrier frequency that is matched to $\omega_0$. The light reflected off M1 and M2 of cavity $j$ drives the mirrors M1 and M2 of cavity $j+1$, respectively. The two outputs of the system are the light reflected off the mirrors of cavity 5. Let $a_j$ describes the oscillator mode of cavity $j$. Working in the rotating frame of reference with respect to $\omega_0$, the dynamics of this system can be described in annihilation-operator form \eqref{eq:qsde-anni} with 

{\small \arraycolsep=3pt\def\arraystretch{1.2}
\begin{align*}
 &F = \left[\begin{array}{ccccc}
                  -\gamma & 0 & 0 & 0 & 0 \\ 
                  -2\gamma & -\gamma & 0 & 0 & 0 \\ 
                  -2\gamma & -2\gamma & -\gamma & 0 & 0 \\
                  -2\gamma & -2\gamma & -2\gamma & -\gamma & 0 \\
                  -2\gamma & -2\gamma & -2\gamma & -2\gamma & -\gamma
                 \end{array}\right],  \quad
 G = \left[\begin{array}{cc}
                  -\sqrt{\gamma} & -\sqrt{\gamma} \\
                  -\sqrt{\gamma} & -\sqrt{\gamma} \\
                  -\sqrt{\gamma} & -\sqrt{\gamma} \\
                  -\sqrt{\gamma} & -\sqrt{\gamma} \\
                  -\sqrt{\gamma} & -\sqrt{\gamma}
                 \end{array}\right], \\
\end{align*}
} 
$H = -G^\dagger$, and $K = I$. The frequency response (Bode plot) of the original system is illustrated in Fig. \ref{fig:passive} by black solid lines. Note that the response across the negative frequencies is a mirror image of the response across the positive frequencies (i.e., the responses are symmetric around the origin) because $F,G,H,K$ are real matrices.

In this example, we are interested in obtaining an approximating system with $r=3$. Note that the balanced truncation method in \cite{Nurd14} is not suitable here because the product of the controllability and observability Gramians is identity, leading to equal Hankel singular values of 1. The eigenvalue truncation method in \cite{Pet12} (truncating subsystems corresponding to larger eigenvalues) is also not suitable because the poles of the system are all at $-10^6$. 
Therefore, we apply the left tangential interpolatory projection model reduction approach. As suggested in Section \ref{subsec:interpol-selection}, we choose the tangent directions to be $(\mu_1,\mu_2,\mu_3) = (e_1,e_1,e_1)$. Note from the system dynamics that choosing $(\mu_1,\mu_2,\mu_3) = (e_1,e_1,e_1)$ or $(\mu_1,\mu_2,\mu_3) = (e_2,e_2,e_2)$ would result in the same subspace ${\cal V}_a$.
Since the input-output responses of the original system have characteristics of both low-pass and high-pass filters, we simplify the optimization by assuming that the interpolation points are of the form $(\sigma_1,\sigma_2,\sigma_3) = (\imath\omega^c,0,-\imath\omega^c)$ (i.e., the interpolation point at 0 matches the low-frequency responses while the other points match the high-frequency responses). 
This form of interpolation points also preserve the symmetric properties of the frequency responses around the origin (i.e., $F_r$, $G_r$, and $H_r$ are real matrices).
We find a local minimizer of the ${\cal H}_2$ optimization problem to obtain the value of $\omega^c$ because the ${\cal H}_\infty$ error is $2$ for any $\omega^c > 0$. We have that $\omega^c = 1.48\times 10^7$. This $\omega^c$ leads to a stable system with the poles at $-5.1541\times 10^5$ and $(-1.0780\pm 0.8142\imath)\times 10^7$.

Fig. \ref{fig:passive} illustrates the frequency response (Bode plot) of the reduced system using red dashed lines. From the figure, we see that the frequency responses of our approximation are closely matched to those of the original system along the pass bands. The error $\|\Xi(s)-\Xi_r(s)\|$ for each $s$ is shown in Fig. \ref{fig:passive2-err} and \ref{fig:passive2-err2}. 
Again, the error is large when $s$ is close to the poles of the systems.
As previously mentioned, the ${\cal H}_\infty$ error is $2.00$. The ${\cal H}_\infty$ error bounds \eqref{eq:passive-left-err-bound} and \eqref{eq:passive-right-err-bound} are both $2.92$. Here, the bounds are conservative because passive quantum systems have bounded real transfer functions, i.e., $\|\Xi(\imath\omega)\| \leq 1$ for all $\omega \in \mathbb{R}$ \cite{MP11}, which implies that the ${\cal H}_\infty$ error $\|\Xi - \Xi_r\|_{\infty} \leq 2$.

\begin{figure}[!t]
 \centering
 \includegraphics[width=0.6\textwidth]{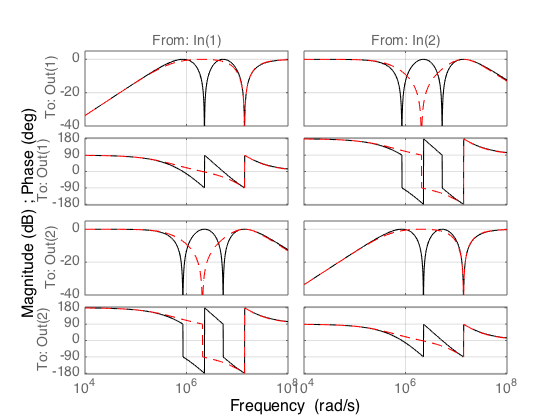}
 \caption{Example \ref{ex:passive2}: The frequency responses (Bode plots) of the original system (black solid line) and the tangential interpolatory projection approximation (red dashed line).}
 \label{fig:passive}
\end{figure}

\begin{figure}[!t]
\centering
\includegraphics[width=0.6\textwidth]{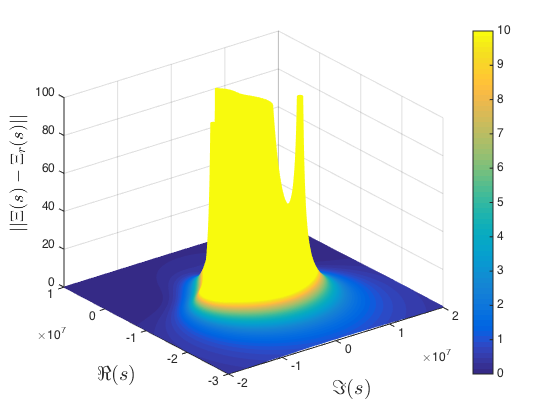}
\caption{Example \ref{ex:passive2}: The approximation error $\|\Xi(s)-\Xi_r(s)\|$.}
\label{fig:passive2-err}
\end{figure}
\begin{figure}[!t]
\centering
\includegraphics[width=0.5\textwidth]{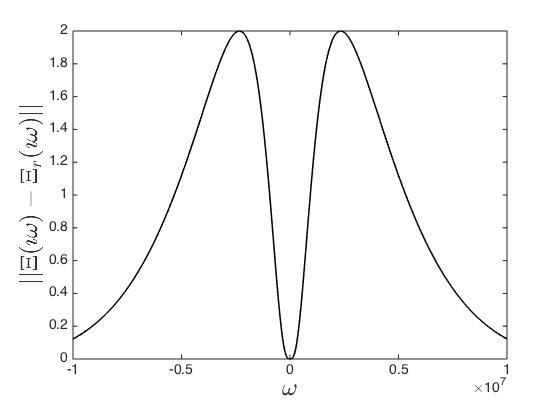}
\caption{Example \ref{ex:passive2}: The approximation error for each purely imaginary point (or frequency) $\|\Xi(\imath\omega)-\Xi_r(\imath\omega)\|$.}
\label{fig:passive2-err2}
\end{figure}
\end{example}

\section{Conclusion} \label{sec:conclusion}

In this paper, we have proposed a tangential interpolatory projection model reduction method for linear quantum stochastic systems. The proposed approach retains the required physical realizability property of the original full-order system while ensuring that the transfer function of the reduced-order system matches that of the original system at multiple points (or frequencies) along some tangent directions. That is, the reduced system can be designed to match specified input-output responses of the original system at various frequencies. We also establish an ${\cal H}_{\infty}$ error bound for the proposed method, formulate optimization based routines to select interpolation points along the imaginary axis, and introduce a heuristic for selecting tangent directions. A passivity preserving model reduction method has also been proposed. We establish a new result illustrating that our reduced-order completely passive linear quantum stochastic system will typically be asymptotically stable. Significantly, our tangential interpolatory projection approach can be applied to a wider class of linear quantum stochastic systems compared to existing model reduction methods for linear quantum systems. Several examples are provided to illustrate the merits of our proposed model reduction approaches.

\small
\bibliographystyle{IEEEtran}
\bibliography{IEEEabrv,rip,otr}

\end{document}